\documentclass[submission,copyright,creativecommons]{eptcs}
\providecommand{\event}{AFL 2026} 

\usepackage{amsthm}
\usepackage{ifthen}
\usepackage{latexsym,amssymb,amsmath,mathdots}
\usepackage{graphics}
\usepackage[dvipsnames]{xcolor}
\usepackage{graphicx}
\usepackage{tcolorbox}
\usepackage{bm}
\usepackage{shuffle}   %
\usepackage{xspace}

\usepackage{tikz}     
\usepackage{tkz-graph}  
\usetikzlibrary{%
  arrows,%
  positioning,%
  decorations.pathmorphing,%
  decorations.pathreplacing,%
  automata,%
  shapes.geometric,%
}

\usepackage{iftex}

\newtheorem{theorem}{Theorem}[section]

\newtheorem{definition}[theorem]{Definition}
\newtheorem{lemma}[theorem]{Lemma}

\theoremstyle{definition}
\newtheorem{remark}[theorem]{Remark}

\newtheorem{example}[theorem]{Example}

\newcommand\pnsi{\par\indent}
\newcommand\pnsn{\par\noindent}
\newcommand\pssi{\par\smallskip\indent}
\newcommand\pssn{\par\smallskip\noindent}

\newcommand\pmsn{\par\medskip\noindent}

\newcommand\emdef[1]{\textsf{#1}}

\newcommand\gray[1]{{\color{gray} #1}}

\newlength{\algoindent}   
\renewcommand{\phi}{\varphi}

\newenvironment{fsm}[1][]
	{\begin{center}
	\begin{tikzpicture}[font=\footnotesize,shorten >=1pt,node distance=3.5cm,
	on grid,>=stealth',initial text=,
	every node/.style={align=center},
	every state/.append style={inner sep=1pt},
	every path/.style={->,bend angle=20},
	#1]}
	{\end{tikzpicture}\end{center}}

\newcommand{\N}{\ensuremath{\mathbb{N}}\xspace}

\newcommand{\al}[0]{\ensuremath{\Sigma}\xspace} 
\newcommand{\alstar}[0]{\ensuremath{\al^*}\xspace} 
\newcommand{\alplus}[0]{\ensuremath{\al^+}\xspace} 
\newcommand{\alG}[0]{\ensuremath{\Gamma}\xspace}
\newcommand{\ew}[0]{\ensuremath{\lambda}\xspace}   
\newcommand{\ews}[0]{\ensuremath{\bm{\lambda}}\xspace}   
 
\newcommand\lang{{\cal L}}

\newcommand\rel{\ensuremath{\omega}\xspace}   
\newcommand\po{\ensuremath{\le}\xspace}   
\newcommand\ppx{\ensuremath{\po_{\mathrm{px}}}\xspace}   
\newcommand\psx{\ensuremath{\po_{\mathrm{sx}}}\xspace}   %

\newcommand\spo{\ensuremath{<}\xspace}   
\newcommand\ssx{\ensuremath{\spo_{\mathrm{sx}}}\xspace}   %

\newcommand\aut{\ensuremath{\bm{a}}\xspace}   
\newcommand\auta{\ensuremath{\bm{a}}\xspace}   
\newcommand\autb{\ensuremath{\bm{b}}\xspace}   
\newcommand\autc{\ensuremath{\bm{c}}\xspace}   
\newcommand\autd{\ensuremath{\bm{d}}\xspace}   
\newcommand{\trt}{\ensuremath{\bm{t}}\xspace}    
\newcommand\trs{\ensuremath{\bm{s}}\xspace}      

\newcommand{\tree}{\ensuremath{\hat\tau}\xspace}    
\newcommand{\treef}{\ensuremath{\hat\phi}\xspace}    
\newcommand{\etree}{\ensuremath{\bm\lambda_{\rm t}}\xspace}    

\newcommand{\valf}{\ensuremath{\dot\phi}\xspace}    
\newcommand{\valt}{\ensuremath{\dot\tau}\xspace}    

\newcommand\indepfont{\mathbb}

\newcommand\descrf{\ensuremath{\varphi}\xspace}  
\newcommand\indepf{\ensuremath{\indepfont{P}_{\descrf}}\xspace}
\newcommand{\indep}[1]{\ensuremath{\indepfont{P}_{#1}}\xspace}

\newcommand\indepP{\ensuremath{\indepfont{P}}\xspace}

\title{Independent Languages and Witnesses \\ of Dependence (Non-satisfaction)\footnote{This work was funded by a Discovery Grant of NSERC Canada (grant number RGPIN-2020-05996).}}
\author{Stavros Konstantinidis
\institute{Mathematics and Computing Science,\\
Saint Mary's University, Halifax, NS, Canada}
\email{s.konstantinidis@smu.ca}
}
\def\titlerunning{Witnesses of Non-satisfaction}
\def\authorrunning{S. Konstantinidis}

\begin{document}
\maketitle

\begin{abstract}
We consider the  independent language property defined by the language equation $\phi(X)=\emptyset$, where the expression $\phi$ involves the language variable $X$, constant languages, and standard regular operations including transductions, but not complementation.
This property consists of all languages $L$ satisfying $\phi(L)=\emptyset$. 
Depending on the choice of the operations in $\phi$, the equation defines a broad class of codes, including combinations of standard variable-length codes and error-detecting codes. 
We show that any $\phi$-independence is a J\"urgensen independence, we define what a witness of non-satisfaction of $\phi(X)=\emptyset$ is, for a language $L$, and we show how to compute a witness of non-satisfaction when $L$ is regular. We also discuss the complexity of the problem, showing that the decision version of the problem is PSPACE-complete.
\end{abstract}

\section{Introduction}\label{sec:intro}
A language $L$ over some alphabet $\al$ is called a comma-free code, if  $LL\cap\alplus L\alplus=\emptyset$. What is a witness of non-satisfaction when $LL\cap\alplus L\alplus\not=\emptyset$?  
More generally, if $\phi(X)$ is a language expression involving 
the language variable $X$, constant languages, and standard regular operations including transductions, but not complementation, what is a witness of non-satisfaction for $\phi(L)=\emptyset$?
A flat witness is simply a word in $\phi(L)$. 
In this paper, we investigate the more detailed concept of tree witness of non--satisfaction, which provides witness words in the languages involved in $\phi(L)$.
This topic is related to the satisfaction problem---decide whether $\phi(L)=\emptyset$ providing no witness---which has been studied extensively in the literature (see Section~\ref{sec:sat}).
\pssn
\textbf{Organization and main results.}
The paper is organized as follows. The next section contains a few basic concepts from automata, formal languages, and codes. 
In \underline{Section~\ref{sec:jurg}}, we review some \emph{methods} of defining classes of independent languages---these classes are called independent properties: the methods of partial orders, $n$-ary relations,  the  broad method of J\"urgensen independence, as well as the method of independence expressions $\phi$ which are the main object of this work. We also review quickly the satisfaction  problem (whether a given language satisfies an independent property).
In \underline{Section~\ref{sec:witnesses}}, we define what we mean by a witness of dependence, that is, what a witness is of the non-satisfaction of $\phi(L)=\emptyset$.
A flat witness is any word in $\phi(L)$, but a tree witness tells what $L$-words can be used to make $\phi(L)$ non-empty.
In \underline{Section~\ref{sec:indep}}, we show that every  property definable by an independence expression is a J\"urgensen property.
In \underline{Section~\ref{sec:invertible:ops}}, we revise standard automata constructions to what we call path-invertible constructions that allow one to compute accepting paths of the automata that were used to produce a constructed automaton.
In \underline{Section~\ref{sec:comp:witness}}, we present the algorithm for computing a tree witness of $\phi(L)\not=\emptyset$, for regular languages $L$. We also  discuss the complexity of the algorithm. If $\phi$ is fixed the complexity is polynomial; otherwise, the decision version of the problem is PSPACE-complete.
In \underline{Section~\ref{sec:last}}, we close with a few concluding remarks and a few directions for future research.

\pnsn\textbf{Abbreviations.}\; 
iff = ``if and only if'',\quad
w.r.t. = ``with respect to''
\quad
BFS = breadth-first search, \quad DFS = depth-first search

\section{Basic Notions and Notation}\label{sec:notation}
We assume the reader to be familiar with basics of formal languages, see e.g., \cite{HopcroftUllman,MaSa:handbook,FLhandbookI}. Some notation:
$\al,\alG$ denote arbitrary alphabets; 
$\bar L$ denotes the complement of the language $L$.
Let $u,v,w,x$ be any words over \al. If $w\in L$ then we say that $w$ is an \emdef{$L$-word}. We use standard operations and notation on words and languages; in particular $|w|$ denotes the length of $w$, and \ew denotes the \emdef{empty word}. 
If $w$ is of the form $uv$ then $u$ is called a \emdef{prefix} of $w$ and $v$ is called a \emdef{suffix} of $w$. If $w$ is of the form $uxv$ then $x$ is called an \emdef{infix} of $w$. 
A prefix $u$ of $w$ is called \emdef{proper}, if $u\not=w$. Similar are the concepts of proper suffix and proper infix.
\pnsi
We use the term \emdef{automaton} for the standard notion of a nondeterministic finite automaton (NFA) with possible \ew-transitions, and the acronym DFA for deterministic finite automaton. 
We  assume the reader to be familiar with basics of  transducers, see e.g., \cite{Be:1979,Sak:2009,Yu:handbook}.
A \emdef{(finite) transducer} is a 6-tuple $\trt=(Q,\al,\alG,E,I,F)$ 
such that $Q$ is the set of states, $I\subseteq Q$ is the nonempty set of start (initial) states, $F\subseteq Q$ is the set of final states,  $\al,\alG$ are the  input and output alphabets, respectively, and $E$ is the finite set of transitions (edges). \emph{In this paper, all transducers have equal input and output alphabets: $\al=\alG$.}
We assume that all transducers are in \emdef{standard form:} in each transition $(p,x/y,q)\in E$, we have that the \emdef{input label} $x$ is either the empty word \ew or a symbol in \al, and the same for the \emdef{output label} $y$.
We view a transducer, or automaton, \trt as a labelled directed graph, so we can talk about a \emdef{path} in \trt.  
The \emdef{label} of a transducer path is the pair $(u,v)$ of words such that $u$ (resp. $v$) is the concatenation of the input (resp. output) labels on the transitions of the path.
A path is \emdef{accepting} if the first state of the path is initial and the last state is final.
The \emdef{relation} $R(\trt)$ \emdef{realized by} \trt is the set of labels $(u,v)$ on the accepting paths of \trt.
The set of outputs of \trt on input $u$ is denoted by $\trt(u)$, that is, $v\in\trt(u)$ iff $(u,v)\in R(\trt)$. Also, for a language $L$, $\trt(L)=\cup_{u\in L}\trt(u)$.
A transducer \trt is called \emdef{input-altering}, if $w\notin\trt(w)$ for all words $w$.
If \trs is an automaton, or transducer, then $|\trs|$ denotes the \emdef{size of \trs} = the sum of the numbers of states and transitions of \trs.

We write $\N,\N_0$ for the sets of positive integers and non-negative integers, respectively. If $S$ is a set then $|S|$ denotes the cardinality of $S$.
For a set $S$, the notation $[S]^n$ is used for the 
\emdef{$n$-fold cross product}  $S\times\cdots\times S$ of $S$. This notation is used to avoid confusion with the $n$-th power $L^n$ of a language  $L$.
\pnsi
A \emdef{language property} is simply a class (set) of languages.
An \emdef{independent property} is any  language property \indepP for which every language $L\in\indepP$ is included in a \indepP-maximal language $M\in\indepP$. 
By a \emdef{\indepP-maximal language} (or maximal language with respect to \indepP) we mean a language $M$ in \indepP such that $M\cup\{w\}\notin\indepP$, for all words $w\notin M$.
With very few exceptions, the various concepts of codes in the literature are captured by the concept of independence.
When $L\in \indepP$ we say that $L$ \emdef{satisfies} the property \indepP, or that $L$ is an \emdef{independent language} w.r.t. \indepP.
Usually, an independent property has a name $\varphi$ and is denoted by $\indepP_{\varphi}$, in which case a language $L$ satisfying $\indepP_{\varphi}$ is also called a \emdef{$\varphi$-independent} language.
As stated in \cite{JuKo:handbook}, the prime examples of independent languages are those studied in the theory of codes. 
Each class of codes is defined by a certain condition on the words of the code (the language). For example
\begin{itemize}
    \setlength{\itemsep}{0pt}%
    \setlength{\parskip}{0pt}%
    \vspace{-0.5\topsep}
	\item A language $L$ is a \emdef{prefix code} if no $L$-word is a  prefix of another $L$-word.
	\item A language $L$ is a \emdef{suffix code} if no $L$-word is a  suffix of another $L$-word.
	\item A language $L$ is an \emdef{infix code} if no $L$-word is an  infix of another $L$-word.
	\item A language $L$ is a \emdef{error-detecting for \trt,} where \trt is an input-altering transducer, if \trt cannot output an $L$-word when an $L$-word is used as input.
	\item A language $L$ is a \emdef{uniquely decodable/decipherable code}, or \emdef{UD-code} for short, if every word can be written in at most one way as a concatenation of $L$-words.
\end{itemize}

\section{J\"urgensen Independence \& Independence Expressions}\label{sec:jurg}
An important direction in the research on independent languages is the investigation of systematic methods that allow one to define and study different classes of these languages. 
Next we review some of these methods.

\pmsn
\textbf{Independent languages via partial orders.}
To our knowledge, the first method for defining code properties is the method of \cite{Shyr:Thierrin:relations} via certain binary relations on words. 
A binary relation \rel is a set of pairs of words, that is, 
$\rel\subseteq\al^*\times\al^*$. 
In \cite{Shyr:book}, the author presents the method via binary relations that are partial orders (in the standard definition). Let \po be a partial order on \alstar.
A language $L$ is called \emdef{\po-independent}, if
$u,v\in L$ and $u\po v$ imply $u=v$.
The same class of languages is defined via the strict version `$<$' of `$\po$':
\begin{equation}\label{eq:spo}
	<\>\cap\>[L]^2=\emptyset,
\end{equation}
where $<\,=\,\po-\{(u,u)\mid u\in\alstar\}$, which is the irreflexive and asymmetric version of `$\po$'.
The classic examples of prefix and suffix codes are defined via the partial orders $\ppx, \psx\subseteq[\alstar]^2$, respectively, or via their strict versions; e.g.,
$\ssx=\{(w,z)\mid z \text{ is a proper suffix of } w\}.$
Thus, $L$ is a suffix code iff no two words $u,v\in L$  are related via $\ssx$.

\pmsn
\textbf{Independent languages via $n$-ary relations.}
In general, for integer $n>0$, an $n$-ary relation \rel over some alphabet \al is a subset of the $n$-fold cross product $[\alstar]^n$.
A language $L$ is called \emdef{\rel-independent}, \cite{JurgYu:1991}, if 
\begin{equation}\label{eq:omega}
	\rel\cap[L]^n = \emptyset.
\end{equation}
A classic example of a class of codes that is defined as an \rel-independence, where \rel is ternary, is the class of \emdef{comma-free codes} \cite{BePeRe:2009,JuKo:handbook}:
$
L \text{ is a comma-free code, if } \alplus L\alplus\cap LL=\emptyset.
$
That is, $\rel\cap[L]^3 = \emptyset$, where $\rel=\{(u,v,w)\mid xwy=uv,\text{ for } x,y\in\alplus,u,v,w\in L\}$.

\pmsn
\textbf{J\"urgensen Independence.} 
Let $n$ be a positive integer or $\aleph_0$ (the cardinality of the natural numbers). An \emdef{$n$-independent property}  is a class of languages \indepP such that
\begin{equation}
L\in\indepP \quad\text{iff}\quad \text{ $L'\in\indepP$ for all $L'\subseteq L$ with $|L'|<n$}.
\end{equation}
A \emdef{J\"urgensen property} is simply an $n$-independent property for some $n\in\N\cup\{\aleph_0\}$.
Thus, to tell whether a language $L$ is independent w.r.t. some \indepP, it is sufficient (and necessary) to ensure that every subset of $L$ with $<n$ elements is independent w.r.t. \indepP. 
For example, the class of prefix codes is a \textbf{3}-independent property because a language $L$ is a prefix code iff every subset of $L$ with at most \textbf{2} words is a prefix code.
The J\"urgensen definition can be used to model  independent properties defined by all previous methods, and 
guarantees that every independent language is included in a maximal one \cite{JuKo:handbook}.

\subsection{Independence via language expressions.}\label{sec:lang:expr}
An \emdef{independence expression} $\phi$ is defined inductively as follows: it is the language variable $X$, or a language constant, or one of $\phi_1\odot\phi_2,\, \phi_1\cap\phi_2,\,(\phi_1)^*,\,  \trt(\phi_1)$, where $\phi_1$ and $\phi_2$ are independence expressions and \trt is any transducer constant. 
The symbol $\odot$ denotes concatenation and, as customary, is usually omitted when writing an expression. On the other hand, this symbol is used when we draw the \emdef{tree of the expression $\phi$}, which is denoted by \treef.
Conversely, if $\tree$ is an  expression tree then $\tau$ is the expression that corresponds to  \tree.
Here is an example of an expression and the corresponding tree
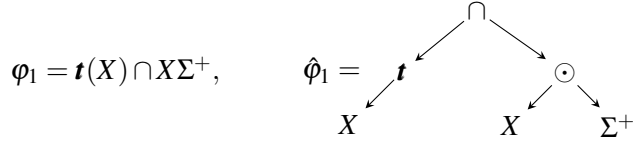
\pmsn\qquad\qquad
\begin{tikzpicture}[>=stealth, shorten >=2pt, auto, node distance=1cm, initial text={}]

\node[inner sep=1pt, minimum size=5pt] (S0){$\cap$};
\node[inner sep=1pt, minimum size=5pt] [below left =.5cm and .7cm of S0](S1){$\trt$};
\node[inner sep=1pt, minimum size=5pt] [below  right =.5cm and .8cm of S0](S2){$\odot$};

\node[inner sep=1pt, minimum size=5pt] [left =0.4cm of S1](tree){$\treef_1=$};
\node[inner sep=1pt, minimum size=5pt] [left =1cm of tree](){$\varphi_1=\trt(X)\cap X\alplus,$};

\node[inner sep=1pt, minimum size=5pt] [below  left of =S1](S5){$X$};
\node[inner sep=1pt, minimum size=5pt] [below  left of =S2](S7){$X$};
\node[inner sep=1pt, minimum size=5pt] [below right of =S2] (S4){$\alplus$};

\path[->]
(S0) edge [swap] node {} (S1)
(S0) edge  node {} (S2)
(S1)   edge [swap] node {} (S5)
(S2)   edge [swap] node {} (S7)
(S2)   edge node {} (S4)
;
\end{tikzpicture}
\pmsn
A nonempty expression tree \tree is also denoted as $(\circ,\tree_1,\ldots,\tree_k)$, where $k\ge0$, $\circ$ is the root of \tree and the $\tree_i$'s are the subtrees of \tree.
For example, the above tree $\treef_1$ is also written as $\big(\cap,(\trt,X),(\odot,X,\alplus)\big)$.
\pssi
If $L$ is a language, we shall write $\valf(L)$ for the language that results by applying the operations involved in $\phi$ on $L$ and on any constant languages occurring in $\phi$.
For example, we write 
$
\valf_1(L)=\trt(L)\cap L\alplus.
$
\pssi
A language $L$ is called \emdef{$\phi$-independent}, if $\valf(L)=\emptyset$.


\pssn
\textbf{About the notation \valf and $\valf(L)$.}
The set that consists of $X$ and constant languages is the basic set in the induction basis for defining an independence expression $\phi$.
The expression can be evaluated when $X$ is replaced with a language $L$, and then the value is denoted by $\valf(L)$.
On the other hand, let $B$ be a (basic) set of objects on which the same operations are well-defined, and let $\phi$ be an  expression based on $B$. Then $\valf$ is the value of $\phi$, which is an element of $B$, or a subset of $B$, depending on the actual type of objects $B$.
For example, if $B=\alstar$,  the set of all words, then $\phi$ is an expression on words which evaluates to a subset\footnote{Using the standard assumption that a word $w$ is treated as the singleton set $\{w\}$.} of $\alstar$---e.g., see tree witness definition in Section~\ref{sec:witnesses}.

\pssn
\textbf{About the absence of the complementation and union operations.} The concept of language expression comes from \cite{Okh:2010}, in the research area of language equations, which allows  the operations of complementation and union to be used in expressions. 
The absence of complementation in the context of independent properties is already noted in \cite{Kon:2017} because it can violate the very fundamental aspect that any subset of a \indepP-property must also be a \indepP-property. 
We have also omitted the union operation for a few reasons: 
(i) Most code properties either can be defined via $\cup$-free expressions, or they are unions\footnote{For example, the property of a binary language being a suffix code and a 1-error detecting language is described by the independence expression $(X\cap X\alplus)\cup(X\cap\trt_0(X))$, where the transducet $\trt_0$ is defined in Example~\ref{ex:witness:tree}. } of $\cup$-free expressions in which case a tree witness of one of these $\cup$-free expressions is sufficient.
(ii) Including the operation would add a few more technical challenges to the goal of computing tree witnesses, which would burden the exposition without offering significant advantages.

\begin{lemma}\label{lem:subset}
Let $\phi$ be an independence expression.
If $L'\subseteq L$ then $\valf(L')\subseteq\valf(L)$.
\end{lemma}
\begin{proof}
	The statement follows easily by the inductive definition of an independence expression, when we note that, for any transducer \trt, $L'\subseteq L$ implies $\trt(L')\subseteq\trt(L)$.
\end{proof}

\subsection{The Satisfaction Problem}\label{sec:sat}
The \emdef{satisfaction problem} for a \emph{fixed} independent property \indepP is to decide whether a given language $L$ satisfies \indepP; that is whether $L\in\indepP$. 
The language is given as input via a well-defined description method. 
Usually, the input is an automaton describing a regular language. 
This is a reasonable assumption in coding theory, as the languages of interest are usually finite, or regular.
Of course it is perfectly fine to assume that the input is a context-free grammar or a pushdown automaton; however in this case, the satisfaction problem is usually undecidable \cite{JuKo:handbook}. 
A large number of references have contributed to the satisfaction problem for finite or regular languages, e.g.,~\cite{Levenshtein:61R,JurgSalYu:1994,JuKo:handbook,sequencesII,Head:Weber:decision,FernReinStai:2007,HanSalomaa:2011,KoHanSalomaa:2021}.
\pnsi
The method of describing independent properties via language expressions is an example of a `formal' method, which 
allows one to study the  \emdef{uniform satisfaction problem}: 
given the description \descrf of an independent property and the description of a language $L$, decide whether $L$ satisfies \indepf.
Two other formal methods for language independence are the method of regular trajectories \cite{Dom:2004} and the method of input-altering transducers \cite{DudKon:2012,KMMR:2018}, both of which are special cases of the method of language expressions.
For example, for an input-altering transducer \trt, \emdef{a language $L$ is \trt-independent} if $\trt(L)\cap L=\emptyset$.\footnote{The transducer method is expressible enough to describe error-detecting languages for many error combinations, including those in~\cite{Kon:2001}. In particular, if \trt is a transducer that flips at least 1 and at most $k$ bits of the input, then a language $L$ is $k$-error detecting iff $\trt(L)\cap L=\emptyset$.} 
Another formal method is the method of trajectory hypersets  \cite{DomSal:2006} which appears to be disjoint from the method of language expressions.

\section{What is a Witness of Dependence/Non-satisfaction?}\label{sec:witnesses}
Consider an independence expression $\phi$ and the property $\indepP_{\varphi}$. If a language $L$ does not satisfy the property then what is a witness of non-satisfaction?

\pssn\textbf{Flat witness.}
A \emdef{flat witness} for $\phi,L$ is any word in $\valf(L)$. As an example, suppose that
\begin{equation}\label{eq:phi:ex}
\varphi_1=\trt(X)\cap X\alplus.
\end{equation} 
Then, a flat witness for $\phi_1,L$ is any word $w\in\trt(L)\cap L\alplus$. 

\pssn\textbf{Tree witness.}
A flat witness $w$ for $\phi_1,L$ satisfies $w\in\trt(w_1)\cap w_2w_3$, for some words $w_1,w_2\in L$ and $w_3\in\alplus$. 
However, $w$ by itself does not give us any information about the two $L$-words $w_1,w_2$ that violate the property $\indepP_{\phi_1}$.
Moreover, the case where the independence expression $\phi$ contains the Kleene star operation should be handled with care.
For example, consider the independence expression
\begin{equation}\label{eq:phi2:ex}
\varphi_2=\trt(X^*)\cap X^*.
\end{equation} 
This is of interest when the variable $X$ represents a UD-code, the input-altering transducer \trt represents a noisy channel, and we want to know whether the set of messages $X^*$ is error-detecting for \trt.
A flat witness for $\phi_2$ and some language $L$ would be a word $w\in\trt(w_1)\cap w_2$ with $w_1,w_2\in L^*$, which gives us no information about which $L$-words are concatenated to get the outcomes $w_1,w_2$.
Moreover, simply picking two words $w_1,w_2\in L$ such that $\trt(w_1^*)\cap w_2^*\not=\emptyset$ does not always work. 
For example, if $\trt$ is input-altering deleting at most 1 bit in any consecutive six bits of a message  and $L=\{02,011,110\}$ then $\trt(w_1^*)\cap w_2^*=\emptyset$ for any $w_1,w_2\in L$, but $\trt(011\,011\,02)\cap(110\,110)\not=\emptyset$. 
For this reason, we define below a more informative  concept of  witness which, 
from a technical convenience point of view, is given in terms of the tree of the expression. 
Note that witnesses of non-satisfaction are computed in \cite{KMMR:2018} for languages that are not \trt-independent as well as for languages that are not UD-codes. 
These types of witnesses are  simpler than what we discuss here.

\begin{definition}\label{def:witness:tree} 
	Let $\phi$ be an independence expression and let $L$ be a language that does not satisfy $\indepP_{\phi}$. 
	A \emdef{tree witness} for $\phi,L$ is an expression tree \tree 
	that evaluates to a \emph{nonempty} subset of $\valf(L)$, that is, 
	$\emptyset\not=\valt\subseteq\valf(L)$, 
	and \tree results from \treef  as follows:
	\begin{enumerate}
    \setlength{\itemsep}{1pt}%
    \setlength{\parskip}{0pt}%
    \vspace{-0.7\topsep}
		\item visit all the internal nodes `*' of \treef in BFS mode, and for each node `*': either replace the node and its subtree with a leaf containing the empty word $\ew$; or replace the node with `$\odot$' and its subtree $\treef_1$, say,  with one or more copies of $\treef_1$ (so now the node should have one or more identical subtrees);\; then,
		\item  visit each leaf $X$  and replace it with an $L$-word;
		\item visit each constant language  leaf $C$ and replace  it with a $C$-word.
	\end{enumerate}
	The (finite) \emph{set of $L$-words used in step~2 above is denoted by $W_{\tree}$}.
	The expression $\tau$ that corresponds to the tree witness is called an \emdef{expression witness for} $\phi,L$.   
\end{definition}

\begin{remark}\label{rem:tree:wit}
	If $\phi$ is a \emdef{basic expression}, that is, $\phi$ is simply the language variable $X$ or a language constant $C$ then \treef consists of just the root ($X$ or $C$) which is also a leaf. If  $C,L\not=\emptyset$ then any tree witness \tree for $\phi,L$ consists of just a root $w$ (with $w\in X$ or $w\in C$) which is also a leaf, and is also a flat witness.
\end{remark}

\begin{remark}\label{rem:tree:wit:correct}
	Let $\phi$ be an independence expression and let $L$ be a language. 
	If $\valf(L)\not=\emptyset$, that is $L$ does not satisfy $\indepP_{\phi}$, then there is indeed a tree witness \tree for $\phi,L$. 
	While this can be shown rigorously via induction on the structure of \treef, the main idea is based on the customary practice of finding  $L$-words for $\valf(L)\not=\emptyset$---see Example~\ref{ex:witness:tree} below.  
\end{remark}

\begin{example}\label{ex:witness:tree}
Let $\al=\{0,1\}$ and $C=00\alstar$.
	Consider the expression $\phi_3=(XX)^*\cap\,\trt_0\big(C\,\trt_0(X)\big)$ and the language $L=00^*1$, where 
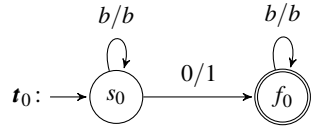
\begin{fsm}[node distance=2.2cm, every state/.style={inner sep=1pt,minimum size=0.7cm}]
	\node [state,initial] (q0) {$s_0$};
	\node [node distance=0.75cm,left=of q0,anchor=east] {$\trt_0\colon$};
	\node [state,accepting,right of=q0] (q1) {$f_0$};
	\path (q0) edge [loop above] node [above] {$b/b$} ()
		(q0) edge node [above] {$0/1$} (q1)
		(q1) edge [loop above] node [above] {$b/b$} ()	
		;
\end{fsm}	
the transducer $\trt_0$ changes exactly one 0 bit of an input word to the bit 1; for example, $\trt_0(1100)=\{1110,1101\}$.\footnote{It is interesting to note that a language $L$ is 1-error detecting if and only if it is $\trt_0$-independent, that is, it is sufficient to detect one 0/1 error  to ensure that $L$ can detect any 1-bit error. 
    }
The word 01010101 is a flat witness for $\phi_3, L$.
	Below we show: the tree $\treef_3$ on the left and  a tree witness $\tree_3$ for $\treef_3,L$ on the right.
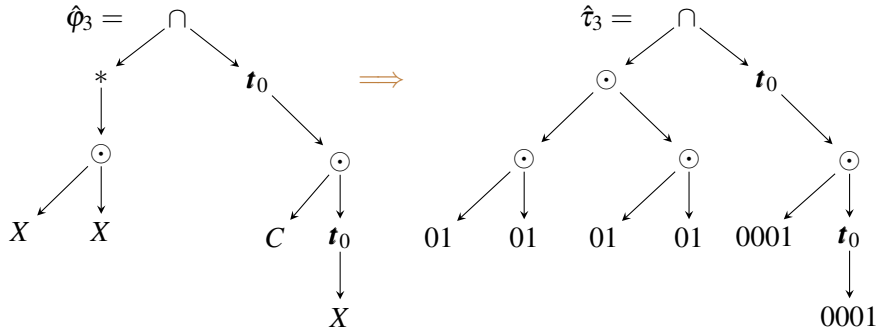
\pmsn
\qquad
\begin{tikzpicture}[>=stealth, shorten >=2pt, auto, node distance=1cm, initial text={}]

\node[inner sep=1pt, minimum size=5pt] (S0){$\cap$};
\node[inner sep=1pt, minimum size=5pt] [below left =.5cm and .7cm of S0](S1){$*$};
\node[inner sep=1pt, minimum size=5pt] [below  right =.5cm and .7cm of S0](S2){$\trt_0$};
\node[inner sep=1pt, minimum size=5pt] [left =0.5cm of S0](){$\treef_3=$};

\node[inner sep=1pt, minimum size=5pt] [below = 0.7cm of S1](doto){$\odot$};
\node[inner sep=1pt, minimum size=5pt] [below right = of S2] (dotth){$\odot$};

\node[inner sep=1pt, minimum size=5pt] [below  left =0.7cm and 0.7cm of doto](w1){$X$};
\node[inner sep=1pt, minimum size=5pt] [below  = 0.7cm of doto](w2){$X$};

\node[inner sep=1pt, minimum size=5pt] [below  = 0.65cm of  dotth] (trt){$\trt_0$};
\node[inner sep=1pt, minimum size=5pt] [below left = 0.65cm and 0.5cm of dotth] (w5){$C$};
\node[inner sep=1pt, minimum size=5pt] [below = 0.7cm of trt] (w6){$X$};

\node[right = of S2, node distance=1.7cm, brown](){$\Longrightarrow$};

\path[->]
(S0) edge [swap] node {} (S1)
(S0) edge  node {} (S2)
(S1)   edge [swap] node {} (doto)
(S2)   edge node {} (dotth)

(doto) edge node {} (w1)
(doto) edge node {} (w2)

(dotth)   edge node {} (trt)
(dotth)   edge node {} (w5)
(trt)   edge node {} (w6)
;
\end{tikzpicture}
\begin{tikzpicture}[>=stealth, shorten >=2pt, auto, node distance=1cm, initial text={}]

\node[inner sep=1pt, minimum size=5pt] (S0){$\cap$};
\node[inner sep=1pt, minimum size=5pt] [below left =.5cm and .7cm of S0](S1){$\odot$};
\node[inner sep=1pt, minimum size=5pt] [below  right =.5cm and .7cm of S0](S2){$\trt_0$};
\node[inner sep=1pt, minimum size=5pt] [left =0.5cm of S0](){$\tree_3=$};

\node[inner sep=1pt, minimum size=5pt] [below left = of S1](doto){$\odot$};
\node[inner sep=1pt, minimum size=5pt] [below right = of S1](dott){$\odot$};
\node[inner sep=1pt, minimum size=5pt] [below right = of S2] (dotth){$\odot$};

\node[inner sep=1pt, minimum size=5pt] [below  left =0.7cm and 0.7cm of doto](w1){$01$};
\node[inner sep=1pt, minimum size=5pt] [below  = 0.7cm of doto](w2){$01$};

\node[inner sep=1pt, minimum size=5pt] [below  left =0.7cm and 0.7cm of dott](w3){$01$};
\node[inner sep=1pt, minimum size=5pt] [below  = 0.7cm of dott](w4){$01$};

\node[inner sep=1pt, minimum size=5pt] [below  = 0.65cm of  dotth] (trt){$\trt_0$};
\node[inner sep=1pt, minimum size=5pt] [below left = 0.65cm and 0.5cm of dotth] (w5){$0001$};
\node[inner sep=1pt, minimum size=5pt] [below = 0.7cm of trt] (w6){$0001$};

\path[->]
(S0) edge [swap] node {} (S1)
(S0) edge  node {} (S2)
(S1)   edge [swap] node {} (doto)
(S1)   edge node {} (dott)
(S2)   edge node {} (dotth)

(doto) edge node {} (w1)
(doto) edge node {} (w2)
(dott) edge node {} (w3)
(dott) edge node {} (w4)

(dotth)   edge node {} (trt)
(dotth)   edge node {} (w5)
(trt)   edge node {} (w6)
;
\end{tikzpicture}
\pssn
We have that $W_{\tree_3}=\{01,0001\}$ and 
the expression witness corresponding to the tree witness is 
the string
$\big((0101)(0101)\big)\cap\,\trt_0\big(0001\,\trt_0(0001)\big)$, where we have omitted $\odot$'s. 
\end{example}


\section{Expression Independence $\subsetneq$ J\"urgensen Independence}\label{sec:indep}
Here we show that, for every independence expression $\phi$, the property $\indep{\phi}$ is a J\"urgensen property.
This implies that every $\phi$-independent language is included in a maximally $\phi$-independent language \cite{JuKo:handbook}.

\begin{lemma}\label{lem:tree:subset}
	Let $\phi$ be an independence expression and let $L$ be a language with $\valf(L)\not=\emptyset$ (that is $L$ does not satisfy $\indepP_{\phi}$). If \tree is a tree witness for $\phi,L$ 
	then
	$\valt\subseteq \valf(W_{\tree})$.
\end{lemma}
\begin{proof} 
The informal argument is that the leafs of the tree witness \tree represent a specific arrangement of the words on which $\phi$ is evaluated, whereas $\valf(W_{\tree})$ represents the set of the evaluations of $\phi$ on the arrangements of words from the set $W_{\tree}$, which include the specific arrangement of the words in the leafs of \tree. 
The statement can be shown by induction on the structure of $\treef$. We show a few parts of the induction.
If $\treef$ is a basic expression tree (as in Remark~\ref{rem:tree:wit}) then the statement follows. 
Now suppose that $\treef$ is a composite expression tree that has a root (which is an operation) and one or more subtrees $\treef_1,\ldots$ such that the statement holds for the tree witnesses $\tree_i$ for $\phi_i,L$, that is, $\valt_i\subseteq\valf_i(W_{\tree_i})$.
We use Lemma~\ref{lem:subset} and the fact that $W_{\tree_i}\subseteq W_{\tree}$.
\pnsi
If the root is `$\cap$' then the tree is $\tree=(\cap,\tree_1,\tree_2)$ and $\valt=\valt_1\cap\valt_2$, so 
$$\valt\subseteq\valf_1(W_{\tree_1})\cap \valf_2(W_{\tree_2})
\subseteq \valf_1(W_{\tree})\cap \valf_2(W_{\tree})=\valf(W_{\tree}),$$
as required. 
Similarly, we can settle the cases where the root is `$\odot$' or a transducer \trt. 
The last case is when the root is `*': then, $\treef$ is of the form $(*,\treef_1)$.
If $\tree=(\ew)$ then $\valt=\{\ew\}=(W_{\tree})^*=\valf(W_{\tree}),$ as required.
If $\tree=(\odot,\tree_1,\ldots,\tree_n)$, for some $n\ge1$, then all $\tree_i$'s are identical except possibly for their leafs, and  similarly each $\tree_i$ is identical to $\treef_1$ except for the leafs.
Then,
\[
\valt=(\valt_1\odot\cdots\odot\valt_n)
\subseteq\big(\valf_1(W_{\tree_1})\odot\cdots\odot \valf_1(W_{\tree_n})\big)
\subseteq\big(\valf_1(W_{\tree})\big)^n
\subseteq\big(\valf_1(W_{\tree})\big)^*=\valf(W_{\tree}),
\]
as required.
\end{proof}

\begin{theorem}\label{th:indep}
	Let  $\phi$ be an independence expression. Then, the property $\indepP_{\phi}$ is a J\"urgensen property. Moreover, if $\phi$ contains no Kleene star operation then $\indepP_{\phi}$  is an $n$-independent property for some $n<\aleph_0$.
\end{theorem}
\begin{proof}
	First we show that $L\in\indepP_{\phi}$ if and only if $L'\in\indepP_{\phi}$ for all finite $L'\subseteq L$.
	The `only if' part follows from Lemma~\ref{lem:subset}.
	For the `if' part, assume that $\valf(L')=\emptyset$ for all finite subsets $L'\subseteq L$, but suppose for the sake of contradiction that $\valf(L)\not=\emptyset$.
	Lemma~\ref{lem:tree:subset} implies that there is a tree witness \tree for $\phi,L$ such that\, $\emptyset\not=\valt\subseteq\valf(L')$ for some finite $L'\subseteq L$. Hence, $\valf(L')\not=\emptyset$, a contradiction!
	\pnsi
	Now suppose that $\phi$ contains no Kleene star operation. 
	We show that there is $n<\aleph_0$ such that $L\in\indepP_{\phi}$ if and only if $L'\in\indepP_{\phi}$ for all finite $L'\subseteq L$ with $|L'|<n$.
	Let $n$ = 1 plus the number of leafs in \treef.
	The `only if' part follows again from Lemma~\ref{lem:subset}.
	For the `if' part assume that the premise holds, but $L\notin\indepP_{\phi}$, that is, $\valf(L)\not=\emptyset$. 
	Definition~\ref{def:witness:tree} implies that there is a tree witness \tree for $\phi,L$ that has an identical structure as the tree \treef (the only difference being that the  leafs of \treef correspond to word leafs of \tree). By Lemma~\ref{lem:tree:subset}, we have
	$\emptyset\not=\valt\subseteq\valf(L')$ for some $L'\subseteq L$ with $|L'|<n$, which contradicts the premise.
\end{proof}

\begin{remark}
	There are uncountably many J\"urgensen properties (see \cite{DudKon:2012}), but only countably many independence expressions. Therefore, there are J\"urgensen properties that are not expressible via any independence expression.
\end{remark}

\section{Path invertible versions of automata constructions}\label{sec:invertible:ops}
We consider the standard language operations 
$
\odot,\; \cap,\; *,\; \trt, 
$
where \trt is any transducer, as well as the known corresponding operations (constructions) on automata, \cite{Yu:handbook,EspBlo:2023}.
We refine those automata constructions in a way that, if we apply an operation on some given automata $\auta_1,\ldots,\auta_k$  and the resulting automaton \auta has an accepting path $\bm p$, say, then the path $\bm p$ can be used to extract accepting paths of the automata $\auta_1,\ldots,\auta_k$. 
The labels (words) on these paths can be used to compute  witness words of $\lang(\auta)\not=\emptyset$.

\pssn
\textbf{Concatenation operation $\odot$.} 
Consider automata $\auta_1,\ldots,\auta_k$, for some $k\ge2$. The automaton $\auta=\auta_1\odot\cdots\odot\auta_k$ consists (as customary) of one copy of each $\auta_i$ with the following specifications:
(i) each state $q$ of some $\auta_i$ is copied as $(i,q)$ in \auta; 
(ii) the start state of \auta is $(1,s)$, where $s$ is the start state of $\auta_1$;  
(iii) the final states of \auta are the states $(k,f)$, for each final state $f$ of $\auta_k$; 
(iv) in addition to the transitions corresponding to those of the $\auta_i$'s, \auta also has, for all $i$, the transitions $\big((i,f),\ew,(i+1,s)\big)$ that connect the copy of every final state $f$ of $\auta_i$ to the copy of the start state $s$ of $\auta_{i+1}$. 
The automaton \auta accepts the language $\lang(\auta_1)\cdots\lang(\auta_k)$. 
\pnsn
\emph{\underline{Inversion rule}:} given any accepting path $\bm p$ of \auta having some label $w$, if we remove the transitions $\big((i,f),\ew,(i+1,s)\big)$, we can effectively compute  accepting paths $\bm p_1,\ldots,\bm p_k$ of $\auta_1,\ldots,\auta_k$ having some labels $w_1,\ldots,w_k$ such that $w=w_1\cdots w_k$.

\pssn
\textbf{Star operation $*$.} 
Consider automaton $\auta_1$. The automaton $\auta=\auta_1^*$ consists (as customary) of one copy of  $\auta_1$ with the following specifications: 
(i) each state $q$ of  $\auta_1$ is copied as $(1',q)$ in \auta;
(ii) there is a new state $(0',0)$ in \auta which is the start state;
(iii) for each copy $(1',f)$ of a final state $f$ of $\auta_1$, we add in \auta the state $(2',f)$;
(iv) the states $(0',0)$ and $(2',f)$ are the final states of \auta, for all $f$ final in $\auta_1$;
(v) in addition to the transitions corresponding to those of  $\auta_1$, \auta also has the following transitions: 
$\big((0',0),\ew,(1',s)\big)$, where $s$ is the start state of $\auta_1$; 
$\big((1',q),a,(2',f)\big)$ for each transition $(q,a,f)$ of $\auta_1$ with $f$ final; and 
$\big((2',f),\ew,(1',s)\big)$ for each final state $f$ of $\auta_1$.
The automaton \auta accepts the language $\lang(\auta_1)^*$.
Technically, the prime in $(1',q)$ in states is not necessary, but it helps as a reminder that the construction is about the `*' operation.
\pnsn
\emph{\underline{Inversion rule}:}  given an accepting path $\bm p$ of \auta having some label $w$, we can look at all occurrences of states $(2',f)$ in $\bm p$ and effectively compute accepting paths $\bm p_1,\ldots,\bm p_n$ of $\auta_1$, for some $n\ge0$, such that the concatenation of the labels of those paths is equal to $w$.
In the special case where $\bm p$ is simply the empty path starting at state $(0,0)$, we have that $n=0$ and $w=\ew$.

\begin{example}\label{ex:invertible1}
Consider the following automaton $\auta_1$ accepting the language  $00^*1$.\begin{fsm}[node distance=1.6cm, every state/.style={inner sep=1pt,minimum size=0.7cm}]
	\node [state,initial] (s) {$s$};
	\node [node distance=0.75cm,left=of s,anchor=east] {$\auta_1\colon$};
	\node [state,right of=s] (q) {$q$};
	\node [state,accepting,right of=q] (f) {$f$};
	\path 
		(s) edge node [above] {$0$} (q)
		(q) edge [loop above] node [right] {$\;0$} ()	
		(q) edge node [above] {$1$} (f)
		;
\end{fsm}	
The automaton $(\auta_1\odot\auta_1)^*$ is shown next, where we abbreviate a simple pair $(p,q)$ as $pq$.
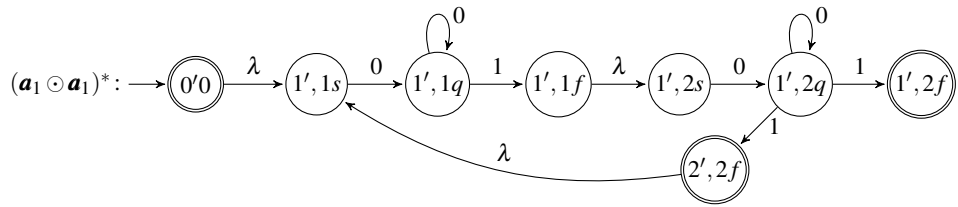
\pssn
\begin{fsm}[node distance=1.6cm, every state/.style={inner sep=1pt,minimum size=0.7cm}]
	\node [state,initial,accepting] (zz) {$0'0$};
	\node [state,right of=zz] (s) {$1',1s$};
	\node [node distance=0.75cm,left=of zz,anchor=east] {$(\auta_1\odot\auta_1)^*\colon$};
	\node [state,right of=s] (q) {$1',1q$};
	\node [state,right of=q] (f) {$1',1f$};

	\node [state,right of=f] (t) {$1',2s$};
	\node [state,right of=t] (r) {$1',2q$};
	\node [state,accepting,right of=r] (g) {$1',2f$};
	\node [state,accepting,below left of=r] (g2) {$2',2f$};
	\path 
		(zz) edge node [above] {$\ew$} (s)
		(s) edge node [above] {$0$} (q)
		(q) edge [loop above] node [right] {$\;0$} ()	
		(q) edge node [above] {$1$} (f)
		(f) edge node [above] {$\ew$} (t)
		(t) edge node [above] {$0$} (r)
		(r) edge [loop above] node [right] {$\;0$} ()	
		(r) edge node [above] {$1$} (g)
		(r) edge node [right] {$1$} (g2)
		(g2) edge [bend left] node [above] {$\ew$} (s)
		;
\end{fsm}	
\end{example}

\pnsn
\textbf{Transducer  operation \trt.} 
Consider a fixed, but arbitrary, transducer \trt and an automaton $\auta_1$. 
We can construct an automaton $\auta=\trt(\auta_1)$ accepting the language $\trt\big(\lang(\auta_1)\big)$ using a standard product construction\footnote{In practice, one does product constructions incrementally, starting from the start state $(s_1,s_2)$ and adding new states and transitions as long as each new state is reached 
from a previously added one.}:
(i) the states of \auta are pairs $(q_1,q_2)$ of states from $\trt,\auta_1$, respectively; 
(ii) $(s_1,s_2)$ is the start state of \auta, where $s_1,s_2$ are the start states of $\trt,\auta_1$, respectively;
(iii) $(f_1,f_2)$ are the final states of \auta, where $f_1,f_2$ are any final states of $\trt,\auta_1$, respectively; and 
(iv) the transitions of \auta are of the form  $\big((p_1,p_2),b,(q_1,q_2)\big)$, when $(p_1,a/b,q_1)$, $(p_2,a,q_2)$ are transitions of $\trt,\auta_1$, respectively\footnote{{Product constructions}\label{foot} are usually based on matching equal (parts of) labels in the transitions of the two machines. However, if the machines involved have  transitions with label \ew then, before the product construction begins, one adds to each state a self-loop with label \ew (for automaton states), or with label $\ew/\ew$ (for transducer states). 
}. 
This construction, however, cannot give us a witness word $w$ that would be used as input to \trt as in the expression~\eqref{eq:phi:ex} because the automaton records only the output of the transducer. 
\emph{For this reason, we make the following simple modification to the construction of \auta:}  its transitions are of the form  $\big((p_1,p_2),ab,(q_1,q_2)\big)$, when $(p_1,a/b,q_1)$, $(p_2,a,q_2)$ are transitions of $\trt,\auta_1$, respectively.
In general, if we need to apply a transducer operation \trs, say, to \auta that has some labels $\beta b$, where $\beta$ could be absent, then the transitions in the product machine would be of the form $\big((p_1,p_2),\beta bc,(q_1,q_2)\big)$, when $(p_2,\beta b,q_2)$, $(p_1,b/c,q_1)$  are transitions of $\auta,\trs$, respectively.
This guarantees that the output of repeated transducers is preserved, as well as the  input to the initial transducer that lead to that output. In case the matching label $b$ (or $c$) above happens to be \ew, in the product label we use the special symbol \ews to record that the matching came from \ew: $\beta \ews$.
Of course  labels $\beta b$ of the resulting machines are not in accordance with the transducer labels defined in Section~\ref{sec:notation}. However, we can \emph{treat the machine as an ordinary automaton if each transition label $\beta b$ is read as $b$}.
\pnsn
\emph{\underline{Inversion rule}:} given any accepting path $\bm p$ of \auta with some label $\beta w$, we can effectively compute an accepting path of $\auta_1$ with label~$\beta$ by changing each transition $\big((p_1,p_2),\beta_1 b,(q_1,q_2)\big)$ in $\bm p$ to $(p_1,\beta_1,q_1)$.
\pssn
\textbf{Intermediate note.} The transition labels $\beta b$ produced by a transducer operation cause no changes to the previous constructions for $\odot,*$, as all transition labels in the constructed automaton are copies of transition labels in the automata used as operands to these operations. 
On the other hand, for the below case of the $\cap$-construction, the labels need to be treated a little carefully. 

\begin{example}\label{ex:invertible2}
Consider the automaton $\auta_1$ in Example~\ref{ex:invertible1}, the transducer $\trt_0$ in Example~\ref{ex:witness:tree}, and the  automaton $\auta_2$ accepting the language  $00\alstar$:
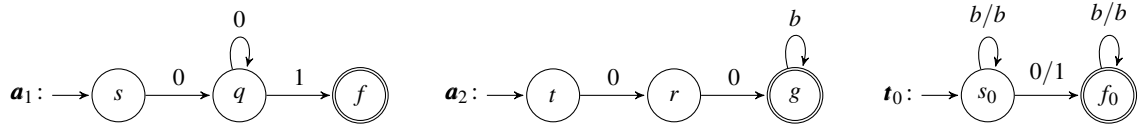
\begin{fsm}[node distance=1.6cm, every state/.style={inner sep=1pt,minimum size=0.7cm}]
	\node [state,initial] (s) {$s$};
	\node [node distance=0.75cm,left=of s,anchor=east] {$\auta_1\colon$};
	\node [state,right of=s] (q) {$q$};
	\node [state,accepting,right of=q] (f) {$f$};

	\node [state,node distance=2.55cm,initial,right of=f] (t) {$t$};
	\node [node distance=0.75cm,left=of t,anchor=east] {$\auta_2\colon$};
	\node [state,right of=t] (r) {$r$};
	\node [state,accepting,right of=r] (g) {$g$};
	\node [state,initial,node distance=2.55cm,right of=g] (q0) {$s_0$};
	\node [node distance=0.75cm,left=of q0,anchor=east] {$\trt_0\colon$};
	\node [state,accepting,right of=q0] (q1) {$f_0$};
	\path 
		(s) edge node [above] {$0$} (q)
		(q) edge [loop above] node [above] {$0$} ()	
		(q) edge node [above] {$1$} (f)
		(t) edge node [above] {$0$} (r)
		(g) edge [loop above] node [above] {$b$} ()	
		(r) edge node [above] {$0$} (g)
        (q0) edge [loop above] node [above] {$b/b$} ()
		(q0) edge node [above] {$0/1$} (q1)
		(q1) edge [loop above] node [above] {$b/b$} ()	
		;
\end{fsm}	
The automaton $\auta_2\odot\trt_0(\auta_1)$ is shown next, where we abbreviate again a simple pair $(p,q)$ as $pq$.
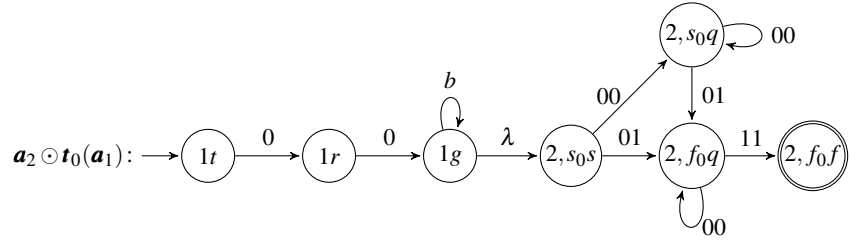
\pssn
\begin{fsm}[node distance=1.6cm, every state/.style={inner sep=1pt,minimum size=0.7cm}]
	\node [state,node distance=2.55cm,initial] (t) {$1t$};
	\node [node distance=0.75cm,left=of t,anchor=east] {$\auta_2\odot\trt_0(\auta_1)\colon$};
	\node [state,right of=t] (r) {$1r$};
	\node [state,right of=r] (g) {$1g$};

	\node [state,right of=g] (s) {$2,s_0s$};
	\node [state,right of=s] (q1) {$2,f_0q$};
	\node [state, above of=q1] (q) {$2,s_0q$};
	\node [state,accepting, right of=q1] (f) {$2,f_0f$};

	\path 
		(t) edge node [above] {$0$} (r)
		(g) edge [loop above] node [above] {$b$} ()	
		(r) edge node [above] {$0$} (g)
		(g) edge node [above] {$\ew$} (s)
		(s) edge node [left] {$00$} (q)
		(s) edge node [above] {$01$} (q1)
		(q) edge [loop right] node [right] {$00$} ()	
		(q) edge node [right] {$01$} (q1)
		(q1) edge [loop below] node [right] {$00$} ()	
		(q1) edge node [above] {$11$} (f)
		;
\end{fsm}	
\end{example}

\pssn
\textbf{Intersection operation $\cap$.} 
Consider automata $\auta_1,\auta_2$. The automaton $\auta=\auta_1\cap\auta_2$ is obtained via the standard product construction, taking into account footnote~\ref{foot} and that automata labels are of the form $\beta b$ due to possible previous transducer operations
:
(i) the states of \auta are pairs $(q_1,q_2)$ of states from $\auta_1,\auta_2$, respectively; 
(ii) $(s_1,s_2)$ is the start state of \auta, where $s_1,s_2$ are the start states of $\auta_1,\auta_2$, respectively;
(iii) $(f_1,f_2)$ are the final states of \auta, where $f_1,f_2$ are any final states of $\auta_1,\auta_2$, respectively; and 
(iv) the transitions of \auta are of the form  $\big((p_1,p_2),[\beta_1,\beta_2]b,(q_1,q_2)\big)$, when $(p_1,\beta_1b,q_1)$, $(p_2,\beta_2b,q_2)$ are transitions of $\auta_1,\auta_2$, respectively\footnote{If $\beta_1$ is blank then the product label is $[,\beta_2]b$, and similarly if $\beta_2$ is blank. If both $\beta_1,\beta_2$ are blank then the label is $[,]b$---see also Remark~\ref{rem:nested:brackets}.}. 
The automaton \auta accepts the language $\lang(\auta_1)\cap\lang(\auta_2)$---recall here, the component $b$ of  labels $\beta b$ is what determines the language of the automaton, as the part $\beta$ of the label simply records previous transducer inputs that resulted into $b$.
\pnsn
\emph{\underline{Inversion rule}:}  given any accepting path $\bm p$ of \auta having some label $[\beta_1,\beta_2]w$, we can  compute  two accepting paths $\bm p_1,\bm p_2$ of $\auta_1,\auta_2$ with labels $\beta_1 w$ and $\beta_2 w$, respectively. Thus, the word $w$ that is accepted in the path $\bm p$ is also the word accepted in the two paths $\bm p_1,\bm p_2$.

\begin{remark}\label{rem:nested:brackets}
	The use of the brackets `[' and `]' in the  labels for `$\cap$' could result into nested such brackets, which is necessary when the path inversion of the  $\cap$-operation is performed. For example, suppose we need to perform the  operations $\trt_1,\trt_2,\cap,\cap$  in the order implied by the expression 
	$
	\big(\trt_1(\auta)\cap\autb\big)\cap\trt_2(\autc),
	$
	where $\auta,\autb,\autc$ are automata and $\trt_1,\trt_2$ are transducers.
	The operation $\trt_1(\auta)$ produces labels of the form $ab$, where $a/b$ and $a$ are labels in the paths of $\trt_1,\auta$, respectively; 
	and the operation $\trt_2(\autc)$ produces labels of the form $cb$, where $c/b$ and $c$ are labels in the paths of $\trt_2,\autc$, respectively. 
	Then, the operation $\trt_1(\auta)\cap\autb$ produces labels of the form $[a,]b$ by matching the labels $ab$ and $b$ of $\trt_1(\auta)$ and $\autb$, respectively.
	In the final automaton \autd, say, produced by the operation `$\cap$' between $\trt_1(\auta)\cap\autb$ and $\trt_2(\autc)$, the labels are of the form $[[a,],c]b$. Thus, an accepting path $\bm p$ of \autd has a sequence of $n$, say, transition labels $[[a_i,],c_i]b_i$ such that the word accepted by $\bm p$ is $b_1\cdots b_n$.
	Using the above inversion rules, we can use the complete information in the labels $[[a_i,],c_i]b_i$ to compute words $w_1,w_2,w_3$ accepted by $\auta,\autb,\autc$, respectively, such that $b_1\cdots b_n\in \big(\trt_1(w_1)\cap w_2\big)\cap\trt_2(w_3)$.
	Without the nesting of brackets, we would not be able to invert~$\bm p$.
\end{remark}

\pnsi
The next statement follows from the  definitions of, and the comments about the above four path invertible operations.

\begin{lemma}\label{lem:invert}
	If each of the above four operations  results into an automaton \auta that has an accepting path $\bm p$ then the inversion rule for $\bm p$ produces some accepting paths $\bm p_1,\ldots,\bm p_k$ of the automata that were used as operands in the said operation.
	Moreover, the cost of producing those paths is $O(|\bm p|)$.
\end{lemma}

\begin{example}\label{ex:path:inv}
	Consider the expression $\phi_3=(XX)^*\cap\,\trt_0\big(C\,\trt_0(X)\big)$, with $C=00\alstar$, and the language $L=00^*1$.
	Below we show an accepting path of the automaton \auta resulting after applying the operations in $\phi_3$ on the automata $\auta_1,\auta_2$  in Example~\ref{ex:invertible2}, which accept $L,C$, respectively.
Again, 
we abbreviate a simple pair $(p,q)$ as $pq.$
\begin{align*}
\big(0'0,(s_0,1t)\big) &\stackrel{[,\ews]\ews}{\longrightarrow}
\big((1',1s),(s_0,1t)\big) \stackrel{[,0]0}{\longrightarrow}
\big((1',1q),(s_0,1r)\big) \stackrel{[,0]1}{\longrightarrow}
\\
\big((1',1f),(f_0,1g)\big) &\stackrel{[,\ews]\ews}{\longrightarrow}
\big((1',2s),(f_0,1g)\big) \stackrel{[,0]0}{\longrightarrow} 
\big((1',2q),(f_0,1g)\big) \stackrel{[,1]1}{\longrightarrow}
\\
\big((2',2f),(f_0,1g)\big) &\stackrel{[,\ews]\ews}{\longrightarrow}
\big((1',1s),(f_0,(2,s_0s))\big)\stackrel{[,00]0}{\longrightarrow}
\\
\big((1',1q),(f_0,(2,s_0q))\big) &\stackrel{[,01]1}{\longrightarrow}
\big((1',1f),(f_0,(2,f_0q))\big)\stackrel{[,\ews]\ews}{\longrightarrow}
\big((1',2s),(f_0,(2,f_0q))\big) \stackrel{[,00]0}{\longrightarrow}
\\
\big((1',2q),(f_0,(2,f_0q))\big) &\stackrel{[,11]1}{\longrightarrow}
\big((2',2f),(f_0,(2,f_0f))\big)
\end{align*}
The above path results from the product of the automata $(\auta_1\odot\auta_1)^*$ and $\trt_0\big(\auta_2\odot\trt_0(\auta_1)\big)$ and can be decomposed into two accepting paths of those automata.
\pssi
Accepting path of $(\auta_1\odot\auta_1)^*$:
\begin{align*}
0'0     &\stackrel{\ew}{\longrightarrow}
(1',1s) \stackrel{0}{\longrightarrow}
(1',1q) \stackrel{1}{\longrightarrow}
(1',1f) \stackrel{\ew}{\longrightarrow}
(1',2s) \stackrel{0}{\longrightarrow} 
(1',2q) \stackrel{1}{\longrightarrow}
(2',2f) \stackrel{\ew}{\longrightarrow}
\\
(1',1s) &\stackrel{0}{\longrightarrow} 
(1',1q)  \stackrel{1}{\longrightarrow}
(1',1f) \stackrel{\ew}{\longrightarrow}
(1',2s) \stackrel{0}{\longrightarrow}
(1',2q) \stackrel{1}{\longrightarrow}
(2',2f) 
\end{align*}
\qquad
\qquad
\parbox{0.8\textwidth}{\small
The above path can be decomposed into two accepting paths of the automaton $(\auta_1\odot\auta_1)$:
\pssi 
$
1s\stackrel{0}{\longrightarrow}
1q\stackrel{1}{\longrightarrow}
1f\stackrel{\ew}{\longrightarrow}
2s\stackrel{0}{\longrightarrow}
2q\stackrel{1}{\longrightarrow}
2f$\quad and \quad
$
1s\stackrel{0}{\longrightarrow}
1q\stackrel{1}{\longrightarrow}
1f\stackrel{\ew}{\longrightarrow}
2s\stackrel{0}{\longrightarrow}
2q\stackrel{1}{\longrightarrow}
2f
$
\pssn
Each of those can be further decomposed into two paths of $\auta_1$:
\pssi\qquad\qquad
$
s\stackrel{0}{\longrightarrow}
q\stackrel{1}{\longrightarrow}
f,
$ \quad 
$
s\stackrel{0}{\longrightarrow}
q\stackrel{1}{\longrightarrow}
f$
\qquad and \qquad  
$s\stackrel{0}{\longrightarrow}
q\stackrel{1}{\longrightarrow}
f,$ \quad 
$
s\stackrel{0}{\longrightarrow}
q\stackrel{1}{\longrightarrow}
f$,
\pssn resulting into the words 01, 01, 01, 01 whose concatenation 01010101 is in $(LL)^*$.
}
\pssi
Accepting path of $\trt_0\big(\auta_2\odot\trt_0(\auta_1)\big)$:
\begin{align*}
(s_0,1t) &\stackrel{\ews\ews}{\longrightarrow}
(s_0,1t) \stackrel{00}{\longrightarrow}
(s_0,1r) \stackrel{01}{\longrightarrow}
(f_0,1g) \stackrel{\ews\ews}{\longrightarrow}
(f_0,1g) \stackrel{00}{\longrightarrow} 
(f_0,1g) \stackrel{11}{\longrightarrow}
(f_0,1g) \stackrel{\ews\ews}{\longrightarrow}
\\
\big(f_0,(2,s_0s)\big)&\stackrel{000}{\longrightarrow}
\big(f_0,(2,s_0q)\big)\stackrel{011}{\longrightarrow}
\big(f_0,(2,f_0q)\big)\stackrel{\ews\ews}{\longrightarrow}
\big(f_0,(2,f_0q)\big) \stackrel{000}{\longrightarrow}
\\
\big(f_0,(2,f_0q)\big) &\stackrel{111}{\longrightarrow}
\big(f_0,(2,f_0f)\big)
\end{align*}
\qquad
\qquad
\parbox{0.8\textwidth}{\small
The above path can be decomposed into an accepting path of the automaton $\big(\auta_2\odot\trt_0(\auta_1)\big)$:
\begin{align*}
1t &\stackrel{\ew}{\longrightarrow}
1t \stackrel{0}{\longrightarrow}
1r \stackrel{0}{\longrightarrow}
1g \stackrel{\ew}{\longrightarrow}
1g \stackrel{0}{\longrightarrow} 
1g \stackrel{1}{\longrightarrow}
1g \stackrel{\ew}{\longrightarrow}
\\
(2,s_0s)&\stackrel{00}{\longrightarrow}
(2,s_0q) \stackrel{01}{\longrightarrow}
(2,f_0q) \stackrel{\ew}{\longrightarrow}
(2,f_0q) \stackrel{00}{\longrightarrow}
(2,f_0q) \stackrel{11}{\longrightarrow}
(2,f_0f)
\end{align*}
This path can be further decomposed into two paths of $\auta_2$ and $\trt_0(\auta_1)$:
\pssn\;
$
t \stackrel{\ew}{\longrightarrow}
t \stackrel{0}{\longrightarrow}
r \stackrel{0}{\longrightarrow}
g \stackrel{\ew}{\longrightarrow}
g \stackrel{0}{\longrightarrow} 
g \stackrel{1}{\longrightarrow}
g 
$ 
$\>\>$ and $\>\>$
$
s_0s \stackrel{00}{\longrightarrow}
s_0q \stackrel{01}{\longrightarrow}
f_0q \stackrel{\ew}{\longrightarrow}
f_0q \stackrel{00}{\longrightarrow}
f_0q \stackrel{11}{\longrightarrow}
f_0f
$,
\pssn
and the latter path is decomposed into a path of $\auta_1$: 
\pssi\qquad
$
s \stackrel{0}{\longrightarrow}
q \stackrel{0}{\longrightarrow}
q \stackrel{\ew}{\longrightarrow}
q \stackrel{0}{\longrightarrow}
q \stackrel{1}{\longrightarrow}
f
$,
\pssn 
resulting into the words $w_5=0001\in C$ and $w_6=0001\in L$ such that $\trt_0\big(w_5\,\trt_0(w_6)\big)\not=\emptyset$.
}
\end{example}

\section{Computing Witnesses of Non-satisfaction}\label{sec:comp:witness}
Here we assume that the constant languages occurring in an independence expression $\phi$ are automata and that \emph{the variable $X$ ranges over automata}. This is well defined, as all the operations involved in $\phi$ can also be `overloaded' to work on automata---see Section~\ref{sec:invertible:ops}. 

\pssn
\textbf{Computing a flat witness.}
If $\auta_1$ is any automaton then $\valf(\auta_1)$ is the automaton resulting if we evaluate $\phi$ on $\auta_1$ using the constructions in Section~\ref{sec:invertible:ops}.
Then, we can either compute\footnote{E.g., via a shortest path algorithm from the initial to a final state of the automaton.} and return a 
word $w\in\lang\big(\valf(\auta_1)\big)$, or return that no such word exists.
This gives a flat witness for $\phi,\lang(\auta_1)$ by consulting the automaton $\valf(\auta_1)$.

\begin{example}\label{ex:flat:wit}
Consider the independence expression $XX\cap\alplus X\alplus$ for describing comma-free codes and the expression tree on the left of the below figure, where $\auta_1$ accepts some language $L$ and $\auta_2$ accepts \alplus. The automaton $\auta_1\odot\auta_1$ accepts $LL$ and the automaton $\auta_2\odot\auta_1\odot\auta_2$ accepts $\alplus L\alplus$.   
To get the final automaton \auta, we use the operation $\cap$ between the automata $\auta_1\odot\auta_1$ and $\auta_2\odot\auta_1\odot\auta_2$.

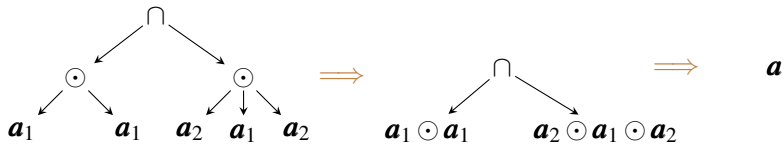
\pmsn
\qquad
\begin{tikzpicture}[>=stealth, shorten >=2pt, auto, node distance=1cm, initial text={}]

\node[inner sep=1pt, minimum size=5pt] (S0){$\cap$};
\node[inner sep=1pt, minimum size=5pt] [below left =.5cm and .7cm of S0](S1){$\odot$};
\node[inner sep=1pt, minimum size=5pt] [below  right =.5cm and .8cm of S0](S2){$\odot$};

\node[inner sep=1pt, minimum size=5pt] [below  left of =S1](S5){$\auta_1$};
\node[inner sep=1pt, minimum size=5pt] [below  right of =S1](S3){$\auta_1$};
\node[inner sep=1pt, minimum size=5pt] [below  left of =S2](S7){$\auta_2$};
\node[inner sep=1pt, minimum size=5pt] [below  =.4cm of S2](S6){$\auta_1$};
\node[inner sep=1pt, minimum size=5pt] [below right of =S2] (S4){$\auta_2$};

\node[right of =S2, node distance=1.3cm,brown](){${\Longrightarrow}$};

\path[->]
(S0) edge [swap] node {} (S1)
(S0) edge  node {} (S2)
(S1)   edge [swap] node {} (S5)
(S1)   edge node {} (S3)
(S2)   edge [swap] node {} (S7)
(S2)   edge node {} (S4)
(S2)   edge node {} (S6)
;
\end{tikzpicture}
\begin{tikzpicture}[>=stealth, shorten >=2pt, auto, node distance=1cm, initial text={}]

\node[inner sep=1pt, minimum size=5pt] (S0){$\cap$};
\node[inner sep=1pt, minimum size=5pt] [below left =.5cm and .2cm of S0](S1){$\auta_1\odot\auta_1$};
\node[inner sep=1pt, minimum size=5pt] [below  right =.5cm and .2cm of S0](S2){$\auta_2\odot\auta_1\odot\auta_2$};

\node[right of =S0, node distance=2.3cm,brown](rarr){${\Longrightarrow}$};
\node[right of =rarr, node distance=1.3cm](){$\auta$};

\path[->]
(S0) edge [swap] node {} (S1)
(S0) edge  node {} (S2)
;
\end{tikzpicture}
\pssn
\end{example}


\subsection{Augmented automata and the Tree Witness Algorithm}\label{sec:augmented}
Referring to Example~\ref{ex:flat:wit}, we would like to be able to start from an accepting path $\bm p$ of \auta and decompose it into two paths $\bm p_3,\bm p_4$ that are also accepting paths of $\auta_1\odot\auta_1$ and $\auta_2\odot\auta_1\odot\auta_2$, respectively; then decompose $\bm p_3$ into two paths that are both accepting paths of $\auta_1$; and then decompose $\bm p_4$ into three paths that are accepting paths of $\auta_2,\auta_1,\auta_2$, respectively.
Then, these paths can be used to retrieve witness words.
This process is described below and uses the concept of an augmented automaton.
Informally, an augmented automaton is a pair $(\auta,\tree)$ such that \auta is an automaton that results after repeatedly applying  operations on some initial list of automata, and $\tree$ is the expression tree that corresponds to the application of these operations. By deconstructing the tree $\tree$ and the states  of the  automaton \auta, it is possible  to infer the operations that led to its construction. 

\begin{definition}\label{def:augm:aut}
	An \emdef{augmented automaton}  is one of the following
	\begin{itemize}
    \setlength{\itemsep}{2pt}%
    \setlength{\parskip}{0pt}%
    \vspace{-0.7\topsep}
		\item A pair $(\auta,\etree)$, where \auta is an automaton and \etree is  the empty tree.
		\item If $(\auta_1,\tree_1),(\auta_2,\tree_2)$ are augmented automata then $\big(\auta_1\cap\auta_2,(\cap,\tree_1,\tree_2))$ is an augmented automaton.
		\item If $k\ge2$ and $(\auta_i,\tree_i)$, for $i=1,\ldots k$, are augmented automata then $\big(\auta_1\odot\cdots\odot\auta_k,(\odot,\tree_1,\ldots,\tree_k)\big)$ is an augmented automaton.
		\item If $(\auta_1,\tree_1)$ is an augmented automaton then $\big(\auta_1^*, (*,\tree_1)\big)$ is an augmented automaton.
		\item If $(\auta_1,\tree_1)$ is an augmented automaton then $\big(\trt(\auta_1), (\trt,\tree_1)\big)$ is an augmented automaton.
	\end{itemize}
\end{definition}

\begin{remark}\label{rem:augmented:aut}
Let $\phi$ be an independence expression and let $\auta_1$ be an automaton. 
If we traverse  the expression tree \treef in DFS mode and evaluate $\phi$ on $\auta_1$ together with the tree \tree according to Definition~\ref{def:augm:aut} then we will get the augmented automaton $\big(\valf(\aut_1),\tree\big)$.
The trees \treef and \tree have identical structures (except that their leafs are different).
\end{remark}

\begin{figure}[ht]
\qquad
\parbox{0.95\textwidth}
{
\hspace*{0.3\algoindent} 
$\mathsf{TreeWitness}\, (\hat\varphi,\auta_1)$
\textsf{
\pssn \hspace*{\algoindent}
Do a DFS traversal of $\hat\varphi$ and compute the augmented automaton $(\auta,\tree)$
\\ \hspace*{\algoindent}
if \auta has no accepting path
\\\hspace*{1.7\algoindent} 
return  \textsf{None}
\\ \hspace*{\algoindent}
\gray{\# below the algorithm assigns a path to each node of \tree}
\\ \hspace*{\algoindent}
\gray{\# and alters \tree when a node `*' is encountered}
\\ \hspace*{\algoindent}
assign an accepting path of \auta to the root of $\tree$
\\ \hspace*{1.0\algoindent} 
Do a BFS traversal of $\tree$ such that when a node $x$ is visited: 
\\\hspace*{1.7\algoindent} 
if $x$ is not the root:
\\ \hspace*{2.7\algoindent}
let $y$ be the parent of $x$
\\ \hspace*{2.7\algoindent}
assign to $x$ the path resulting from the path of $y$
\\ \hspace*{2.7\algoindent}
\gray{\# the above is based on the path inversion rules  in Sec. \ref{sec:invertible:ops}}
\\ \hspace*{1.7\algoindent}
if $x$ is a * node:
\\ \hspace*{2.7\algoindent}
if the path $\bm p$ of $x$ contains no state of the form $(2',f)$:
\\ \hspace*{3.7\algoindent}
replace $x$ and its subtree with the empty word \ew (a leaf)
\\ \hspace*{2.7\algoindent}
else:
\\ \hspace*{3.7\algoindent}
let $n\ge1$ be the number of states of the form $(2',f)$ in $\bm p$;
\\ \hspace*{3.7\algoindent}
replace the subtree of $x$ with $n$ identical copies of the subtree
\\ \hspace*{3.7\algoindent}
and assign to them the $n$ paths resulting from $\bm p$
\\ \hspace*{3.7\algoindent}
\gray{\# again the above is based on the path inversion rules in Sec. \ref{sec:invertible:ops}}
\\ \hspace*{\algoindent} 
Remove all paths from the internal nodes of \tree 
\\ \hspace*{\algoindent} 
\gray{\# below the algorithm replaces the path of each leaf with a word} 
\\ \hspace*{\algoindent} 
For each leaf node $x$ that is not \ew: 
\\ \hspace*{1.7\algoindent} 
replace the path of $x$ with the word (label) of the path
\\ \hspace*{\algoindent} 
Return the tree \tree  
}  
}  
\end{figure}
\begin{theorem}\label{th:tree:witness}
Let $\phi$ be an independence expression and let $\auta_1$ be an automaton.
If $\lang\big(\valf(\auta_1)\big)\not=\emptyset$ then $\mathsf{TreeWitness}\,(\treef,\auta_1)$ returns a tree witness  for $\phi,\lang(\auta_1)$.
\end{theorem}
\begin{proof} 
	First note that after the initial DFS traversal, the produced tree \tree  has exactly the same structure as \treef (the only difference being that each leaf of \tree is \etree, whereas each leaf of \treef is $X$ or a language constant), and that the produced automaton \auta accepts $\lang\big(\valf(\auta_1)\big)$.
	In the BFS traversal, the structure of the tree \tree is altered when the visited node $x$ is `*', which agrees with what \tree should be according to Definition~\ref{def:witness:tree}.
	Moreover, after the BFS traversal, every leaf node $x$ has a path assigned to it, unless $x$ is \ew as a result of substituting \ew for a subtree with root `*'.
	The label of that path is used to set the leaf to a word, as required.
	If $\valf\big(\lang(\auta_1)\big)\not=\emptyset$, the resulting tree \tree satisfies $\emptyset\not=\valt\subseteq\valf\big(\lang(\auta_1)\big)$ because: at the root level, the algorithm computes an accepting path of the automaton \auta which accepts $\valf\big(\lang(\auta_1)\big)$, and at each lower level tree node $x$, the parent path is used to produce \emph{accepting paths} of the automata corresponding to the operands of the operation $x$, according to Lemma~\ref{lem:invert} about the path inversion rules in Section~\ref{sec:invertible:ops}.
\end{proof}

\subsection{Complexity}\label{sec:cxty}
The algorithm might use exponential space, as the expressions $\phi$ could involve an unbounded number of $\cap$-operations and each operation $\auta\cap\autb$ takes time $O(|\auta|\cdot|\autb|)$.
In fact the decision version of the problem, which is the uniform satisfaction problem of whether $\valf(\auta_1)$ accepts at least one word is PSPACE-complete---see below Theorem~\ref{th:cxty}.
What if $\phi$ is considered fixed? 
As stated in Section~\ref{sec:sat}, there have been several results on the  satisfaction problem of fixed properties \indepP. 
If there is a (fixed) expression $\phi$ such that $\indep{\phi}=\indepP$ then the automaton $\valf(\auta_1)$ can be computed in polynomial time:  

\begin{theorem}\label{th:cxty}
The following statements hold true.
\begin{enumerate}
    \setlength{\itemsep}{1pt}%
    \setlength{\parskip}{0pt}%
    \vspace{-0.7\topsep}
	\item For fixed expression tree \treef, algorithm $\mathsf{TreeWitness}(\auta_1)$ works in time $O(|\auta_1|^{1+|\phi|_\cap}),$ where $|\phi|_\cap$ is the (fixed) number of $\cap$-operations occurring in $\phi$.
	\item The uniform satisfaction problem of whether $\lang\big(\valf(\auta_1)\big)=\emptyset$, for given $\phi,\auta_1$, is PSPACE-complete.
\end{enumerate}
\end{theorem}
\begin{proof}
The first statement is a consequence of the  following points: 
(i) the time complexity is the sum of time complexities to perform the operations in the internal nodes of \treef, plus the time complexities to produce the path of each node from the parent node path;
(ii) the operations `*', `$\odot$' and \trt take linear time (as each transducer operation \trt is fixed); 
(iii) the operation `$\cap$' takes at most quadratic time;
(iv) by Lemma~\ref{lem:invert}, computing each node path can be done in linear time. 

\pnsi
For the second statement, first note that the uniform satisfaction problem is PSPACE-hard because the empty DFA intersection problem, which is PSPACE-hard, \cite{FerKre:2017}, can be reduced to it as follows: Any list $\auta_1,\ldots,\auta_k$ of DFAs, can be mapped in polynomial time to the instance $(\treef,\auta_1)$, where $\phi=X\cap\auta_2\cap\cdots\cap\auta_k$.
Now we need to show that the problem is in PSPACE.
It is sufficient to show a nondeterministic algorithm that decides whether $\lang\big(\valf(\auta_1)\big)\not=\emptyset$ using polynomial space.
Let $\ell=1+|\phi|_{\cap}+|\phi|_{\rm tr}$, where $|\phi|_{\rm tr}$ is the number of transducer operations in $\phi$.
Then, the number of states in the automaton $\valf(\auta_1)$ is $\le m^{\ell}$, where $m$ is the largest of the sizes of the automata and transducers occurring in $\treef$. 
This implies that $\lang\big(\valf(\auta_1)\big)\not=\emptyset$ iff $\valf(\auta_1)$ accepts a word of length $< m^{\ell}$. Next we describe the nondeterministic algorithm~$A$.
\pnsi
The main idea is that  algorithm $A$ will first compute the start state of $\valf(\auta_1)$.
Then, $A$ will guess $< m^{\ell}$  symbols. 
For each guessed symbol $\sigma_i$, $A$ will guess the next state of $\valf(\auta_1)$ using only $\sigma_i$ and the current state (thus only the current and next state need to be stored in each iteration). 
To construct the start state, $A$  does a DFS traversal of the tree \treef. When the traversal goes up, from leafs to the root, it adds a component of the start state at each visited tree node,  according to the path invertible operations described in Section~\ref{sec:invertible:ops}.
When the current state $\theta$ has been computed and the symbol $\sigma_i$ has been guessed, the next state $\theta'$ is computed via a DFS traversal of \treef as follows. When a node (other than the root) is visited, the state of the parent node is decomposed, depending on the node operation, and the appropriate decomposed state component is stored at that node. 
When a leaf node is reached, the decomposed state and $\sigma_i$ are used to guess a transition to the next state of the automaton at the leaf node. As in the case of the start state construction, the next state at an internal node is constructed from the constructed next states of the children, according to the path invertible operations described in Section~\ref{sec:invertible:ops}. 
When all guessed letters $\sigma_i$ have been processed, the algorithm returns true if the last constructed state is final in $\valf(\auta_1)$. 
\pnsi
We continue the proof with some further details---see also below Example~\ref{ex:pspace}.  In the path that is being guessed, \emph{each symbol $\sigma_i$ is allowed to be \ew, which would match labels \ew in transitions}.
In the case where the DFS goes down along a subtree whose root is a transducer operation \trt, where \trt is the highest occurrence of such an operation, the algorithm, going further down to a leaf automaton, guesses a transition with some label $a$ and then going up to \trt, a transition with label $a/\sigma$ in \trt is guessed. 
This was when there is no other transducer operation between the leaf and \trt. In general, if there are some $k\ge 1$ intermediate transducer operations, then, after $a$ is guessed at the leaf level, going up, the algorithm guesses: a transition with label $a/a_1$ at the next transducer, another transition with label $a_1/a_2$ at the next transition label, until a guess of a transition with label $a_k/\sigma$ at the level of \trt. All these transition guesses are used to construct the next state $\theta'$.
\end{proof}

%
\begin{example}\label{ex:pspace}
Consider again the expression tree $\treef_3$ in Example~\ref{ex:witness:tree}.	
In the picture below, the left tree shows the build-up  of the \emph{start state} $\big(00',(0,1t)\big)$. Note that when DFS is going down and encounters a *-node, the state of that node is $(0,0')$ and there is no need to visit the subtree below.\footnote{As in Example~\ref{ex:path:inv}, we sometimes abbreviate a simple pair $(x,y)$ as $xy$.} 
In the right part of the left tree, when DFS reaches the leaf $\auta_2$, it finds the start state $t$ and then goes up building the start states of the higher nodes.
Now suppose that the current guessed symbol is $\sigma_i$ and the current state is 
$\theta=\big((1',1x),(i,(2,jy))\big)$, which means that, when $\theta$ was computed, $x$ was the state of the leftmost $\auta_1$, $(1,x)$ was the state of the left $\odot$-node, $(1',1x)$ was the state of the *-node, $y$ was the state of the rightmost $\auta_1$, $(j,y)$ was the state of the low \trt-node, $(2,jy)$ was the state of the right $\odot$-node, and  $\big(i,(2,jy)\big)$ was the state of the high \trt-node.
The right tree shows how the next state $\theta'=\big((b',bx'),(i',(2,j'y'))\big)$ could be computed from $\theta$ and $\sigma_i$. The DFS traversal goes down to the leftmost leaf $\auta_1$, guesses a transition $(x,\sigma_i,x')$, or in case $x$ is a final state in $\auta_1$ and $\sigma_i=\ew$, it could  go up to $\odot$ and then down to the second $\auta_1$ and find its start state $x'=s$. 
Then, the current state of the first $\odot$-node is $(b,x')$, where $b$ is 1 or 2. Then, the state of the *-node is $(b',bx')$, where $b'$ is $1'$ or $2'$, according to the *-construction in Section~\ref{sec:invertible:ops}.
When the DFS traversal goes down the right part of $\treef_3$, it reaches the rightmost leaf $\auta_1$ according to the right part $(i,(2,jy))$ of the current state $\theta$. There it guesses a transition $(y,a,y')$ and, going up, it guesses a transition $(j,a/a_1,j')$ of the low \trt and sets the state of the low \trt-node to $(j',y')$; then, it sets the state of the right $\odot$-node to $(2,j'y')$; then, it guesses a transition $(i,a_1/\sigma_i,i')$ of the high $\trt$ and sets the state of the high \trt node to $(i',(2,j'y'))$. Finally, the state $(b',bx')$ of the *-node and the state $(i',(2,j'y'))$ of the high $\trt$ node are used to make the next state $\theta'$, as required.
\end{example}

\begin{figure}
\qquad
\begin{tikzpicture}[>=stealth, shorten >=2pt, auto, node distance=1cm, initial text={}]

\node[inner sep=1pt, minimum size=5pt] (S0){$\cap$};
\node[inner sep=1pt, minimum size=5pt] [right =.1cm of S0,blue](){\small $\big(00',(s_0,1t)\big)$};
\node[inner sep=1pt, minimum size=5pt] [below left =.5cm and .7cm of S0](S1){$*$};
\node[inner sep=1pt, minimum size=5pt] [left =0.1cm of S1,blue](){\small $00'$};
\node[inner sep=1pt, minimum size=5pt] [below  right =.5cm and .7cm of S0](S2){$\trt$};
\node[inner sep=1pt, minimum size=5pt] [right =.1cm of S2,blue](){\small $(s_0,1t)$};
\node[inner sep=1pt, minimum size=5pt] [left =0.5cm of S0](){$\treef_3=$};

\node[inner sep=1pt, minimum size=5pt] [below = 0.7cm of S1](doto){$\odot$};
\node[inner sep=1pt, minimum size=5pt] [below right = of S2] (dotth){$\odot$};
\node[inner sep=1pt, minimum size=5pt] [left = 0.1cm of dotth,blue] (){\small $1t$};

\node[inner sep=1pt, minimum size=5pt] [below  left =0.7cm and 0.7cm of doto](w1){$\auta_1$};
\node[inner sep=1pt, minimum size=5pt] [below  = 0.7cm of doto](w2){$\auta_1$};

\node[inner sep=1pt, minimum size=5pt] [below  = 0.65cm of  dotth] (trt){$\trt$};
\node[inner sep=1pt, minimum size=5pt] [below left = 0.65cm and 0.5cm of dotth] (w5){$\auta_2$};
\node[inner sep=1pt, minimum size=5pt] [left = 0.06cm of w5,blue] (){\small $t$};
\node[inner sep=1pt, minimum size=5pt] [below = 0.7cm of trt] (w6){$\auta_1$};

\path[->]
(S0) edge [swap] node {} (S1)
(S0) edge  node {} (S2)
(S1)   edge [swap] node {} (doto)
(S2)   edge node {} (dotth)

(doto) edge node {} (w1)
(doto) edge node {} (w2)

(dotth)   edge node {} (trt)
(dotth)   edge node {} (w5)
(trt)   edge node {} (w6)
;
\end{tikzpicture}
\qquad\qquad
\begin{tikzpicture}[>=stealth, shorten >=2pt, auto, node distance=1cm, initial text={}]

\node[inner sep=1pt, minimum size=5pt] (S0){$\cap$};
\node[inner sep=1pt, minimum size=5pt] [right =.1cm of S0,blue](){\small $\big((b',bx'),(i',(2,j'y'))\big)$};
\node[inner sep=1pt, minimum size=5pt] [below left =.5cm and .7cm of S0](S1){$*$};
\node[inner sep=1pt, minimum size=5pt] [left =0.1cm of S1,blue](){\small $(b',bx')$};
\node[inner sep=1pt, minimum size=5pt] [below  right =.5cm and .7cm of S0](S2){$\trt$};
\node[inner sep=1pt, minimum size=5pt] [right =.1cm of S2,blue](){\small $(i',(2,j'y'))$};
\node[inner sep=1pt, minimum size=5pt] [left =0.5cm of S0](){$\treef_3=$};

\node[inner sep=1pt, minimum size=5pt] [below = 0.7cm of S1](doto){$\odot$};
\node[inner sep=1pt, minimum size=5pt] [left =0.1cm of doto,blue](){\small $bx'$};
\node[inner sep=1pt, minimum size=5pt] [below right = of S2] (dotth){$\odot$};
\node[inner sep=1pt, minimum size=5pt] [right =0.1cm of dotth,blue](){\small $(2,j'y')$};

\node[inner sep=1pt, minimum size=5pt] [below  left =0.7cm and 0.7cm of doto](w1){$\auta_1$};
\node[inner sep=1pt, minimum size=5pt,gray] [below =0.1cm of w1](){$x'$ is in one of the two $\aut_1$'s};
\node[inner sep=1pt, minimum size=5pt] [below  = 0.7cm of doto](w2){$\auta_1$};

\node[inner sep=1pt, minimum size=5pt] [below  = 0.65cm of  dotth] (trt){$\trt$};
\node[inner sep=1pt, minimum size=5pt] [right =0.1cm of trt,blue](){\small $j'y'$};
\node[inner sep=1pt, minimum size=5pt] [below left = 0.65cm and 0.5cm of dotth] (w5){$\auta_2$};
\node[inner sep=1pt, minimum size=5pt] [below = 0.7cm of trt] (w6){$\auta_1$};
\node[inner sep=1pt, minimum size=5pt] [right =0.1cm of w6,blue](){\small $y'$};

\path[->]
(S0) edge [swap] node {} (S1)
(S0) edge  node {} (S2)
(S1)   edge [swap] node {} (doto)
(S2)   edge node {} (dotth)

(doto) edge node {} (w1)
(doto) edge node {} (w2)

(dotth)   edge node {} (trt)
(dotth)   edge node {} (w5)
(trt)   edge node {} (w6)
;
\end{tikzpicture}
\end{figure}

\begin{remark}
While the uniform decision problems for the methods of regular trajectories and input-altering transducers are efficiently decidable, it should not be surprising that the uniform satisfaction problem for independence expressions is hard.
In fact, the uniform satisfaction problem for the method of regular trajectory hypersets is undecidable \cite{DomSal:2006}. 
The latter two problems are expressive enough to reduce known hard problems.	
\end{remark}

\section{Concluding Remarks}\label{sec:last}
We have considered independent languages $L$ satisfying the language equation $\phi(X)=\emptyset$. We showed that these languages are independent in the sense of J\"urgensen independence and presented an algorithm that computes a tree witness of $\phi(L)\not=\emptyset$, for given $\phi$ and regular $L$.
The method of independence expressions is probably not powerful enough to express the property of UD-codes, something that we leave as future problem to investigate.
Another problem for future research is the investigation of witnesses of non-satisfaction for independent properties involving combinations of the UD-code and other properties.

Following a suggestion of one referee of this paper, we add as future research the question of constructing a maximal language satisfying a given expression $\phi$, potentially inside a fixed class of languages; for instance: ``is there an infinite prefix code satisfying $\phi$ and, if yes, can we compute a maximal such language ?". We note that, for the case of \trt-independence defined by the expression $\trt(X)\cap X$ (see Section~\ref{sec:sat}), the work of \cite{KonMas:2017} gives a partial answer to this question. 
Another referee suggestion for future research is to investigate the notion of minimal tree witness (using the fewest number of $L$-words and/or shortest such words).

\bibliographystyle{eptcs}
\bibliography{refs.bib}

\end{document}